\documentclass[10pt,reqno]{amsart}
\usepackage{amssymb,mathrsfs,color}
\usepackage{pinlabel}

\usepackage{amsmath, amsfonts, amsthm, verbatim, amssymb}
\usepackage{epstopdf, mathtools}
\mathtoolsset{showonlyrefs}
\usepackage{color}
\usepackage{esint}
\usepackage{stmaryrd}
\usepackage{graphicx}

\usepackage{empheq}
\usepackage[shortlabels]{enumitem}

\newcommand{\rar}{\rightarrow}

\newcommand{\cl}{\mathcal}

\def\grad {{\nabla}}

\newcommand{\bs}[1]{\boldsymbol{#1}}
\newcommand{\ms}[1]{\mathsf{#1}}
\newcommand{\msb}[1]{\boldsymbol{\mathsf{#1}}}

\renewcommand{\div}{\operatorname{div}}

\definecolor{deepgreen}{cmyk}{1,0,1,0.5}

\newcommand{\p}{\partial}

\makeatletter

\newcommand{\Rmnum}[1]{\expandafter\@slowromancap\romannumeral #1@}
\makeatother

\newcommand{\Del}[1]{}

\usepackage[multiple]{footmisc}
\numberwithin{equation}{section}

\newtheorem{thm}{Theorem}[section]
\newtheorem{cor}[thm]{Corollary}

\newtheorem{prop}[thm]{Proposition}

\theoremstyle{remark}

\newtheorem{defn}[thm]{Definition}

\usepackage{tensor}

\usepackage{graphicx}
\usepackage{xcolor} % A package to add color.
\usepackage{tensor}
\usepackage{slashed}
\usepackage{cite}
\definecolor{green}{rgb}{0,0.8,0} % Redefines the color green.
\newcommand{\tr}{\mathrm{tr}\,}

\newcommand{\eps}{\epsilon}

\newcommand{\bbE}{\mathbb E}

\newcommand{\bbR}{\mathbb R}

\newcommand{\tens}{\otimes}

\newcommand{\Sym}{\mathrm{Sym}}

\newcommand{\Div}{\mbox{Div}\,}

\begin{document}

\title[Nonlinear relations between stress and linearized strain]{A mathematical justification for nonlinear constitutive relations between stress and linearized strain}
\author{K. R. Rajagopal and C. Rodriguez}

\begin{abstract} 
We present an asymptotic framework that rigorously generates nonlinear constitutive relations between stress and linearized strain for elastic bodies. Each of these relations arises as the leading order relationship satisfied by a one-parameter family of nonlinear constitutive relations between stress and nonlinear strain. The asymptotic parameter limits the overall range of strains that satisfy the corresponding constitutive relation in the one-parameter family while the stresses can remain large (relative to a fixed stress scale). This differs from classical linearized elasticity where a fixed constitutive relation is assumed, and the magnitude of the displacement gradient serves as the asymptotic parameter. Also unlike classical approaches, the constitutive relations in our framework are expressed as implicit relationships between stress and strain rather than requiring stress explicitly expressed as a function of strain, adding conceptual simplicity and versatility. We demonstrate that our framework rigorously justifies nonlinear constitutive relations between stress and linearized strain including those with density-dependent Young's moduli or derived from strain energies beyond quadratic forms. 
\end{abstract}

\maketitle

\section{Introduction}

\subsection{Nonlinear constitutive relations between stress and linearized strain}
The celebrated eponymous constitutive relation to describe the response of elastic solids subject to ``small strains" remains a cornerstone of solid mechanics.\footnote{Referring to it as ``Hooke's law" is clearly unwarranted as it is merely a constitutive relation that holds for a small subset of bodies in an approximate sense while we expect a ``law" to have more universal validity. Unfortunately, most people that use the constitutive relation, including experts in the field, are unaware that Hooke did not limit the use of the relation to solids.} Hooke did not have a clear grasp of the applicability of the empirical constitutive relation he was proposing in his immortal treatise ``Lectures de potential restitutiva, or of Spring, Explaining the Power of Springy Bodies". He erroneously supposed that the constitutive relation was equally applicable to gases and thereby confused his constitutive relation with the one proposed by Boyle (see \cite{Moyer75} for a discussion of the relevant issues). Moyer \cite{Moyer75} remarks:
\begin{quotation}
\begin{comment}
Hooke also briefly mentions two additional experiments that demonstrate his `theory,' one involving `a straight piece of wood laid Horizontal,' and the other `a body of Air.' Elaborating on this latter pneumatic experiment and referring to his foregoing explicit summary of other basic `trials', Hooke comments:

``The manner of trying the same thing upon a body of Air, whether it be for the rarefaction or for the compression thereof I did about fourteen years since publish in my Micrographia, and therefore I shall not need to add any further description thereof." (p. 335).

Constantly urging the reader throughout these introductory pages to `observe exactly', `examine,' `measure', and `compare', Hooke finally suggests that `it is very evident' that the various experiments all lead to his `Rule or Law of Nature': ``$\ldots$in every springing body$\ldots$the force or power thereof to restore it self to its natural position is always proportionate to the Distance or space it is removed therefrom, whether it be by rarefaction, or separation of its parts the one from the other, or by a Condensation, or crowding of those parts nearer together." (p. 336)
\end{comment}
While it is a simple and straightforward matter to extract from the early pages of De potentia restitutiva the modern statement of `Hooke's law' for elastic solids, $F = -kx$, one must first suppress or ignore Hooke’s `trial' involving `a body of Air.' That is, in modern terms, he inappropriately associated `Hooke's law'-a direct proportionality between displacement and force-with `Boyle's law'-an inverse proportionality between the volume of an ideal gas at constant temperature and the pressure to which it is subjected (or, for a gas in a cylinder with a uniform cross-sectional area, between the length of the column of gas and the force on it).
\end{quotation}

In fact, Moyer \cite{Moyer75} points out that Hooke's theoretical explanation of his empirical relation $F = -kx$ is consistent with Boyle's relation for gases but is \textit{inconsistent} with Hooke's linear relation for solids. Indeed, assuming that a spring has natural length of 1 unit, Hooke proposed that the force due to the spring's particles' ``endeavor of receding from each other" is inversely proportional to its stretched length $d$ \cite{Moyer75}. As a consequence of balance of forces, Moyer \cite{Moyer75} observes:
\begin{quotation}
$\ldots$realizing that the stretched length $d$ is equal to the original length 1 unit plus some increment distance $x$, we obtain a relation between the force on a spring due to an external macroscopic agent and the spring's change in length: 
\begin{align}
	F_e = k (1 - [1/(1+x)]),
\end{align}
which is also equal to $k[x/(1+x)]\ldots$In other words, for slight changes in length, Hooke's model does lead approximately to his empirical law. 
\end{quotation}

In his authoritative article in the Handbuch der Physik, Bell \cite{MechanicsSolidsI} devotes a significant amount of discussion to the work by illustrious contemporaries of Hooke on nonlinear constitutive relations between stress and linearized strain.\footnote{A great deal of care has to be exercised in analyzing the discussion in Bell \cite{MechanicsSolidsI} as some of the works, while concerned with small strains, are suffering inelastic response. As no unloading data is provided, it is possible that during the reported experiments, some or many of the specimens might have undergone inelastic response.} According to Bell, many of them did not believe that Hooke’s empirical constitutive relation was appropriate, even for elastic solids undergoing small displacement gradients.\footnote{Of course, it depends on what one means by small displacement gradients. Here, it means the use of linearized strain wherein the nonlinear part of the strain is ignored and merely the symmetric part of the displacement gradient is used, or put differently, wherein the square of the norm of the displacement gradient is negligible compared to the norm of the displacement gradient. For example, it would be reasonable to say the displacement gradient is small if its Frobenius norm is less than $10^{-2}$. See \cite{Raj_GumMetal14} for more on this point.} Concerning the data that James Bernoulli had sent him, Leibniz observed that it seems to be fitted best by a hyperbola. Later, James Bernoulli himself proposed a parabolic relation, namely
\begin{align}
	t = kx^m,
\end{align}
where $t$ is the applied force and $x$ is the elongation. Bullfinger \cite{Bullfinger1729} proposed the value $m = 3/2$.

Wertheim \cite{Wertheim1847} carried out extensive and systematic experiments on a miscellany of biological materials. He carried out experiments on a variety of bones (femur, fibula), tendons (big toe flexor, small plantar), sartorial muscles, nerves, femoral arteries, veins, for both males and females. With the exception of bones, all his experiments could be fitted to the nonlinear constitutive relation
\begin{align}
	\eps^2 = A \sigma^2 + B \sigma,
\end{align}
where $\eps$ denotes the strain, $\sigma$ the stress, and $A$ and $B$ are constants. This led Bell \cite{MechanicsSolidsI} to remark:
\begin{quotation}
The experiments of 280 years have demonstrated amply for every solid substance examined with sufficient care, that the strain resulting from applied stress is not a linear function thereof.
\end{quotation}
The above observation of Bell remains true today with regard to a large class of solid bodies.

Numerous recent works study a variety of materials wherein one observes nonlinear response even when the body is subject to sufficiently small strains, regimes where the linearized constitutive relation due to Hooke seems reasonable. It is now well established that the response of the trabecular bone is nonlinear, even for small strains, and more importantly such bones respond differently in tension and compression (see, e.g., \cite{Morganetal2001} and \cite{RajetalBones2024}). It is interesting to note that the title of \cite{Morganetal2001} is \emph{Nonlinear behavior of trabecular bone at small strains}, emphasizing that even in the small strain range where the linearized constitutive relation due to Hooke is usually thought to be valid, the response is nonlinear.\footnote{If the strain is sufficiently small and one assumes that the stress as a function of strain is sufficiently smooth that one can use a Taylor series expansion, there will be a small interval containing the origin wherein the linear relationship will hold. However, this might be a absurdly small range of strains. The point being made here is that even in the range wherein one usually applies the linearized approximation, many materials exhibit nonlinear response.} There is considerable experimental evidence for a large class of metallic alloys wherein the response is nonlinear in the small strain regime. For the case of Gum metal and titanium alloys, see \cite{SAITOETAL03, SAKetal04, SAKAGUCHI2005, HAO05, LI2007, TALLING2008669, WITHEY200826, ZHANG2009733}. Once again, we quote a few remarks from the papers which emphasize the fact that the response is nonlinear even though the strain is small. Talling et al. \cite{TALLING2008669} remark: ``As the elastic regions of the microscopic stress–strain curves are nonlinear in Figure 3, the moduli were calculated from the lowest stress portion of the curve, before deviation from linearity.". The first sentence of the abstract of Zhang et al. \cite{ZHANG2009733} reads ``We report the fatigue endurance of a multifunctional b type titanium alloy exhibiting nonlinear elastic deformation behavior". Finally, Withey et al. \cite{WITHEY200826} demarcate the nonlinear response region which is in the small strain regime. Interestingly, even traditional materials such as concrete and rocks exhibit nonlinear behavior in the small strain regime. Grasley et al. \cite{Grasleyetal2015} show that concrete subject to simple uniaxial compression test starts to respond nonlinearly even at a very low strain of $2 \times 10^{-3}$, where one would unhesitatingly apply the classical linearized constitutive relation.

Thus, it is necessary to put into place a framework wherein the linearized strain is a nonlinear function of the stress, or wherein the linearized strain and the stress are implicitly related as an approximation. Moreover, such an approximation should stem rigorously from a proper constitutive relation that meets the basic requirements such as frame-indifference and appropriate material symmetry.

\subsection{Implicit constitutive theory}

We now discuss modern advances in constitutive theory that inspired this study's proposed asymptotic foundation for rigorously generating nonlinear constitutive relations between stress and linearized strain. 

Cauchy's classical theory of elasticity is founded upon a functional relation between the Cauchy stress tensor $\bs T$ and {deformation gradient $\bs F$} for a given deformation,\footnote{Here, $\bs F = \grad_{\bs X} \bs \chi$ is the gradient of a smooth, invertible deformation $\bs \chi: \cl B \rar \bs \chi(\cl B)\subseteq \bbE^3$ of a body with reference configuration $\cl B$ in three-dimensional Euclidean space. See Section 2 for more on the definitions of the tensors appearing in this discussion. Throughout this work, we use standard tensor notation. In particular, if $\bs A$ is a tensor, we denote the Frobenius norm $|\bs A| = [\tr(\bs A \bs A^T)]^{1/2}$.}
\begin{align}
\bs T = \bs g(\bs F). \label{eq:Cauchyelasticity}
\end{align}
Here and throughout our study, we omit the dependence of various quantities on current points in the body, that is, we assume the body to be homogeneous.
Assuming frame indifference, one alternatively prescribes a functional relationship
between the symmetric Piola-Kirchhoff stress tensor $\bar{\bs S}$ and the Green-Saint Venant strain tensor $\bs E = \frac{1}{2}(\bs F^T \bs F - \bs I)$,
\begin{align}
	\bar{\bs S} = \bs f(\bs E). \label{eq:Cauchyelasticityref}
\end{align}

Classic linearized elasticity is obtained from \eqref{eq:Cauchyelasticityref} or \eqref{eq:Cauchyelasticity} in the asymptotic limit of infinitesimal strains. More precisely, if $|\bs F - \bs I| = \delta_0 \ll 1$, $\bs f(\bs 0) = \bs 0$, and $\bs f$ is twice continuously differentiable, then \eqref{eq:Cauchyelasticityref} implies that
\begin{align}
	\bs \sigma = \msb C[\bs \eps] + \bs O(\delta_0^2), \label{eq:linearizedelasticity}
\end{align}
where $\msb C = D_{\bs E}\bs f(\bs 0)$ is the elasticity tensor, $\bs \eps = \frac{1}{2}(\bs F + \bs F^T) - \bs I$ is the linearized strain, and $\bs \sigma$ has replaced $\bar{\bs S}$ as the stress variable.\footnote{We use standard big-oh and little-oh notation. In particular, a tensor $\bs A$ and scalar $B$ satisfy $\bs A = \bs O(B)$ if there exists a constant $C > 0$ such that $|\bs A| \leq C B$. If $C$ depends on parameters $a$ and $b$, we say that the big-oh term $\bs O(B)$ depends on $a$ and $b$.}

In a series of papers \cite{Raj_Implicit03, RajElastElast, Raj2010, RajConspectus}, the first author put into place an implicit framework for the response of elastic bodies, 
\begin{align}
	\bs g(\bs F, \bs T) = \bs 0, \label{eq:rajimplicit}
\end{align}  
with a simpler subclass for isotropic solids taking the form 
\begin{align}
	\bs B = \bs g (\rho/\rho_R, \bs T ) \label{eq:henckycauchy},
\end{align}
where $\bs B = \bs F \bs F^T$ is the left Cauchy-Green stretch tensor. 
Inspired by \cite{Raj_Implicit03, RajElastElast, Raj2010, RajConspectus}, Mai and Walton \cite{MAIWALTONELLIP, MAIWALTONMON} studied constitutive relations of the form 
\begin{align}
	\bs E = \bs f(\bar{\bs S}). \label{eq:maiwalton}
\end{align} 
They proved in \cite{MAIWALTONELLIP} that for a popular class of isotropic forms of $\bs f$ appearing in the literature, strong ellipticity for \eqref{eq:maiwalton} holds as long as $|\bar{\bs S}|$ is sufficiently small and fails for extreme compression. Results of a similar spirit were obtained in \cite{MAIWALTONMON} for a form of monotonicity that is strictly weaker than strong ellipticity (and thus applies when $\bs f$ is not Frech\'et differentiable). It was shown in \cite{PRUSA2020} that certain relations of the form \eqref{eq:henckycauchy} arise when describing the response of isotropic elastic solids specified in terms of Gibbs potentials. For simple choices of the Gibbs potential (see Section 5 of \cite{PRUSA2020}), \eqref{eq:henckycauchy} becomes
\begin{align}
	\bs H = \beta_0(\rho/\rho_R) (\tr \bs T) \bs I + \beta_1(\rho/\rho_R) \bs T, \label{eq:b2zero}
\end{align}
where $\bs H = \frac{1}{2} \log \bs B$ is the \textit{Hencky strain}. Based on the assumption that the Frobenius norm of the displacement gradient is small, the works \cite{Raj_Implicit03, RajElastElast, Raj2010, RajConspectus} formally argued that if an implicit relation $\bs g(\bs B, \bs T) = \bs 0$ holds, then $\bs g(\bs I + 2 \bs \eps, \bs T) = \bs 0$, up to negligible errors of size $|\bs F - \bs I|^2$.\footnote{Similar arguments in \cite{Raj_Implicit03, RajElastElast, Raj2010, RajConspectus} suggest that if the implicit relation $\bs f(\bs E, \bar{\bs S}) = \bs 0$ is satisfied, then $\bs f(\bs \eps, \bar{\bs S}) = \bs 0$, up to negligible errors of size $|\bs F - \bs I|^2$. } 

However, one may show using the implicit function theorem that if $\bs g$ is continuously differentiable, $\bs g(\bs I, \bs 0) = \bs 0$, and $D_{\bs T}\bs g(\bs I, \bs 0)$ is invertible on the linear space of symmetric tensors, then there exists a function $\bs h$ such that \eqref{eq:rajimplicit} is satisfied if and only if $\bs T = \bs h(\bs B)$, provided $|\bs F - \bs I|$ is sufficiently small. A similar statement holds for implicit relations of the form
\begin{align}
	\bs f(\bs E, \bar{\bs S}) = \bs 0. \label{eq:rajimplicit2}
\end{align}
Thus, under the previous assumptions, fixed constitutive relations of the form $\bs g(\bs B, \bs T) = \bs 0$ or \eqref{eq:rajimplicit2} will always lead to \eqref{eq:linearizedelasticity} as the leading order relation satisfied in the asymptotic limit, $\delta_0 = |\bs F - \bs I| \rar 0$. 

Instead, our approach is inspired by the theory discussed in \cite{RajElastElast, RajSmallStrain} wherein strains are \emph{limited by the constitutive relation} to a range of size $\delta \ll 1$ (while the stresses can remain \textit{large} relative to a fixed stress scale). Rather than $\delta_0$, it is the \emph{limiting strain} determined by the constitutive relation, $\delta$, that serves as the asymptotic parameter in our framework. The central thesis demonstrated in this work, is that ``linearization" with respect to $\delta$ rigorously leads to nonlinear relations between stress and linearized strain $\bs \eps$ while, as discussed above, ``linearization" with respect to $\delta_0 = |\bs F - \bs I|$ only leads to linear relations between stress and the linearized strain.\footnote{Here, our use of ``linearization" with respect to a small parameter $a$ means to neglect terms of order $a^2$ appearing in either the strain variable or constitutive relation.}

We illustrate the above discussion and motivate our framework via a simple $1d$ example. Consider the relations 
\begin{gather}
	E = \delta a (1 + |a \bar S|^p)^{-1/p}\bar S, \quad \bar S \in \mathbb R, \label{eq:1dconst} \\
	\eps := -1 + (1 + 2 E)^{1/2}, \quad \delta_0 := |\eps|, \label{eq:1dstrain}
\end{gather} 
where $\eps$ is the linearized strain variable, $\bar S$ is the stress variable, and $$E = \eps + \frac{1}{2} \eps^2$$ is the nonlinear strain. Here, we assume that $\delta, \delta_0 \ll 1$. By \eqref{eq:1dconst}, for all $\bar S \in \mathbb R$, $|E| \leq \delta$. By \eqref{eq:1dstrain}, $\eps = E + O(\delta^2)$, and \eqref{eq:1dconst} yields 
\begin{align}
	\eps + O(\delta^2) = \delta a (1 + |a \bar S|^p)^{-1/p}\bar S. 
\end{align}
Thus, linearization with respect to $\delta$ yields the nonlinear relation between stress and linearized strain, 
\begin{align}
		\eps = \delta a (1 + |a \bar S|^p)^{-1/p}\bar S. \label{eq:1dconstlin}
\end{align}
On the other hand, we may invert \eqref{eq:1dconst}, leading to 
\begin{align}
	a \bar S = (1 -|E/\delta|^p)^{-1/p} E/\delta. \label{eq:1dinvert}
\end{align}
Since $\eps = E + O(\delta_0^2)$, we conclude that 
\begin{align}
a \bar S = \eps/\delta + O(\delta_0^2/\delta^2). \label{eq:delta0approximation}
\end{align}
Thus, linearization with respect to $\delta_0$ yields the linear relation  
\begin{align}
 a \bar S = \eps/\delta. 
\end{align}
This interplay between inversion and linearization of constitutive relations was first explored in \cite{RAJAGOPAL2018}. We point out that it was quite simple to rigorously arrive at \eqref{eq:1dconstlin} directly from \eqref{eq:1dconst} when linearizing with respect to $\delta$. However, it is much less clear how one arrives at $a \bar S = (1-|\eps/\delta|^p)^{-1/p} \eps/\delta$ directly from \eqref{eq:1dinvert} when linearizing with respect to $\delta$. 

\subsection{Main results and outline}
In this work, we consider implicit constitutive relations,
\begin{align}
	\bs E = \bs f_\delta(\bs E, \bar{\bs S}), \label{eq:impcontrelation}
\end{align}
with $\bs f_\delta$ bounded by $\delta \ll 1$ on its domain. The parameter $\delta$ should be physically interpreted as the maximum possible strain the material can undergoe (in the class of processes being considered). In Section 2 we briefly review the kinematics and balance laws necessary for this study to be self-contained. 

In this study, we view each $\bs f_\delta$ as a member of one-parameter family of \emph{strain-limiting functions with limiting small strains} introduced in Section 3. See Definition \ref{dfn:1}. A one-parameter family of strain-limiting functions should be physically interpreted as representing a collection of constitutive relations, with each member modeling a (possibly) distinct material within a broader family of materials, and $\delta$ giving the maximum possible strain the material can undergo (in the class of processes being considered). The defining features of our asymptotic framework are:
\begin{itemize}
	\item We approximate a family of strain-limiting constitutive relations with limiting small strains. 
	\item The limiting strain $\delta$ for each individual constitutive relation in the family acts as the framework's asymptotic parameter characterizing the degree of the approximation, e.g., the leading order approximate.
\end{itemize}

We discuss several examples in Section 3.2, and we show in Section 3.3 that within our framework, a nonlinear constitutive relation between stress and linearized strain is the leading-order-in-$\delta$ relationship satisfied by the family of strain-limiting constitutive relations represented by \eqref{eq:impcontrelation}; informally,
\begin{align}
	\bs \eps = \bs f_\delta(\bs \eps, \bs \sigma) + \bs O(\delta^2).
\end{align}
See Proposition \ref{p:2} and Corollary \ref{c:1} for the precise statements. Our results are clearly distinct but in similar spirit to the situation in classical linearized elasticity \eqref{eq:linearizedelasticity}. As applications of Proposition \ref{p:2} and Corollary \ref{c:1}, we first obtain nonlinear constitutive relations between stress and linearized strain wherein the stress is derived from a strain energy depending on the linearized strain; see \eqref{eq:Greenelasticlin}. In contrast to classical linearized elasticity, however, our framework allows for strain energies beyond quadratic forms of the linearized strain. In addition, we also obtain certain popular nonlinear constitutive relations between stress and linearized strain having density dependent Young's moduli; see \eqref{eq:lindensitydependent1}, \eqref{eq:lindensitydependent2}, \eqref{eq:lindensityconst1}, and \eqref{eq:lindensityconst2}. We emphasize that the nonlinear constitutive relations between stress and linearized strain obtained correspond to the leading order relations satisfied by appropriate families of strain-limiting constitutive relations. 

In the concluding Section 4, we speculate on the interesting open question concerning the solvability of the fully nonlinear problem if one can prove the solvability of the associated ``linearized problem", the body force being kept the same.

\subsection*{Acknowledgments} The authors gratefully acknowledge support provided by NSF Grant DMS-2307562.

\section{Preliminaries}

In this brief section, we give necessary definitions for the ensuing discussion and analysis presented in this work. 

Let $\cl B \subseteq \bbE^3$ be a smooth domain in three-dimensional Euclidean space, the reference configuration of a body with reference mass density $\rho_R$. Let $\bs \chi: \cl B \rar \bs \chi(\cl B) \subseteq \bbE^3$ be a smooth invertible deformation of $\cl B$. For a given reference element $\bs X \in \cl B$, we denote by $\bs x = \bs \chi(\bs X)$ its position in the current configuration $\bs \chi(\cl B)$. The \textit{deformation gradient} of $\bs \chi$ is defined by  
$
\bs F(\bs X) = \frac{\p \bs \chi}{\p \bs X}(\bs X)$.
%Conservation of mass dictates that the mass density on the current configuration is given by 
%\begin{align}
%	\rho = (\det \bs F)^{-1} \rho_R
%\end{align}

The left and right \textit{Cauchy-Green stretch} tensor fields are given respectively by
\begin{align}
	\bs B = \bs F \bs F^{T}, \quad \bs C = \bs F^T \bs F,
\end{align}
and the \emph{Green-Saint Venant strain} is the second order tensor
\begin{align}
	\bs E = \frac{1}{2} (\bs C - \bs I).%, \quad \bs H = \frac{1}{2} \log \bs B. 
\end{align}
%and from the fact that $\det e^{\bs H} = e^{\tr \bs H}$, we have $e^{-\tr \bs H} = \rho/\rho_R$. 
The \textit{displacement field} is defined by 
\begin{align}
	\bs u = \bs x - \bs X,
\end{align}
and the \textit{displacement gradient} is given by 
\begin{align}
	\nabla_{\bs X} \bs u = \bs F - \bs I. %, \quad \nabla_{\bs x} \bs u = \bs I - \bs F^{-1}.
\end{align}
Then 
\begin{gather}
	\bs B = I + \bigl [ \nabla_{\bs X} \bs u + (\nabla_{\bs X} \bs u)^T \bigr ] + \nabla_{\bs X} \bs u (\nabla_{\bs X} \bs u)^T, \\
		\bs C = I + \bigl [ \nabla_{\bs X} \bs u + (\nabla_{\bs X} \bs u)^T \bigr ] + (\nabla_{\bs X} \bs u)^T \nabla_X \bs u.
\end{gather}

Assuming that 
\begin{align}
 |\nabla_{\bs X} \bs u| = \delta_0 \ll 1, \label{eq:infinitesemaldisplacement}
\end{align}
we then have that
\begin{align}
	\bs B = \bs I + 2\bs \eps + \bs O(\delta_0^2), \quad \bs E = \bs \eps + \bs O(\delta_0^2), \label{eq:linstrains}%, \quad \bs H = \bs \eps + \bs O(\delta_0^2), \label{eq:linstrains}
\end{align}
where $\bs \eps$ is the \emph{linearized strain}
\begin{align}
	\bs \eps = \frac{1}{2} \bigl [ \nabla_{\bs X} \bs u + (\nabla_{\bs X} \bs u)^T \bigr ] = \frac{1}{2} (\bs F + \bs F^T) - \bs I.
\end{align}

In the purely mechanical setting and in the absence of external forces and body couples, the classical equations expressing conservation of mass, balance of linear momentum, and balance of angular momentum on the current configuration are given respectively by: 
\begin{gather}
	\dot \rho + \rho \div \bs v = \bs 0, \\
	\rho \dot{\bs v} = \div \bs T, \quad
	\bs T^T = \bs T. 
\end{gather}
Here, $\dot{} = \p_t + \sum_{j = 1}^3 v_j \p_{x_j}$ is the material time derivative, $\rho$ is the current mass density, $\bs v$ is the velocity field, and $\bs T$ is the (symmetric) \textit{Cauchy stress} tensor. On the reference configuration, the balance laws can be expressed by: 
\begin{gather}
	\rho = (\det \bs F)^{-1} \rho_R, \\
	\rho_R \p_t^2 \bs \chi = \mbox{Div}\, (\bs S), \quad \bs S \bs F^T = \bs F \bs S^T,
\end{gather}
where $\bs S$ is the first \textit{Piola-Kirchhoff stress} tensor. We denote the second \textit{symmetric Piola-Kirchhoff stress tensor} by $\bar{\bs S} = \bs F^{-1} \bs S$. We note that if $|\nabla_{\bs X} \bs u| = \delta_0 \ll 1$, then 
\begin{align}
	\rho/\rho_R = (\det \bs B)^{-1/2} = (\det \bs C)^{-1/2} = [\det ( \bs I + 2 \bs E)]^{-1/2} = 1 - \tr \bs \eps + O(\delta_0^2). \quad \label{eq:lindensity}
	%	1 - \rho/\rho_R = 
	%	1 - e^{-\tr \bs H} = \tr \bs \eps + O(\delta_0^2). \label{eq:lindensity}
\end{align}

\section{An asymptotic framework for nonlinear relations between stress and linearized strain}

In this section, we introduce a novel asymptotic framework for constitutive relations. Within this framework, a nonlinear constitutive relation between stress and linearized strain is the leading-order-in-$\delta$ relationship satisfied by a one-parameter family of strain-limiting constitutive relations $\bs E = \bs f_\delta(\bs E, \bar{\bs S})$. These one-parameter families are given in terms of families of strain-limiting functions with limiting small strains, $\bs f_\delta$, defined in the following subsection.

\subsection{Families of strain-limiting functions with limiting small strains}

\begin{defn}\label{dfn:1}
For small $\tilde{\delta} > 0$, and each $\delta \in (0, \tilde \delta)$, let $U_\delta \subseteq B(\bs 0, 1/2)$ and $V$ be open subsets of $\Sym$.\footnote{Here, $\Sym$ denotes the set of symmetric tensors on $\bbR^3$, and $B(\bs 0, r) = \bigl \{ \bs E \in \Sym \mid |\bs E| < r \bigr \}$. We note that if $\bs E \in B(\bs 0, 1/2)$, then it straightforward to see that $\bs C = \bs I + 2 \bs E$ is positive definite.} We say that a collection of bounded Lipschitz continuous functions $\bs f_\delta : U_\delta \times V \rar \Sym$ indexed by $\delta \in (0, \tilde \delta)$ is a family of \textit{strain-limiting functions with limiting small strains} if there exist $C_0, C_1 > 0$ and $D_0 > 0$, independent of $\delta$, such that for all $\delta \in (0,\tilde \delta)$,
\begin{gather}
	\forall \bs E, \bar{\bs S}, \,
	|\bs f_\delta(\bs E, \bar{\bs S})| \leq C_0 \delta, \quad \forall \bs E_1 \not = \bs E_2, \bar{\bs S}, \, \frac{|\bs f_\delta(\bs E_2, \bar{\bs S}) - \bs f_\delta(\bs E_1, \bar{\bs S})|}{|\bs E_2 - \bs E_1|} \leq C_1, \,\,\,\, \label{eq:strainlimiting1}\\ 
	 \forall \bs E, \bar{\bs S}_1 \not = \bar{\bs S}_2, \, \frac{| \bs f_\delta(\bs E, \bar{\bs S}_2) - \bs f_\delta(\bs E, \bar{\bs S}_1)|}{|\bar{\bs S}_2 - \bar{\bs S}_1|} \leq D_0 \delta. \label{eq:strainlimiting2}
\end{gather}
\end{defn}
If $\bs E$ represents a strain variable, $\bar{\bs S}$ represents a stress variable, and $\bs f_\delta$ represents a dimensionless response function, then $\delta, C_0,$ and $C_1$ are dimensionless and $D_0$ has physical units of $(\mbox{length})^2(\mbox{force})^{-1}$. 
Each individual function $\bs f_\delta$ is referred to as a \textit{strain-limiting function with limiting small strains} on $U_\delta \times V$.

We note that $\bs f_\delta$ is a family of strain-limiting functions with limiting small strains if and only if for any fixed $C > 0$, $\bs f_{C^{-1}\delta}$ is a family of strain-limiting functions with limiting small strains. In particular, we may always assume without loss of generality that $C_0 = 1$. However, it will be mathematically convenient to allow $C_0 \neq 1$ when discussing examples in the sequel. 

\begin{comment}
We observe that if $\bs f_\delta : U_\delta \times V \rar \Sym$ is a family of \textit{strain-limiting functions with limiting small strains}, then so is 
\begin{align}
	\tilde{\bs f}_\delta(\bs E, \tilde{\bs S}) = \bs f_{C_0\delta} (\bs E, C_0 D_0^{-1} \tilde{\bs S}), 
\end{align} 
with $(\bs E, \tilde{\bs S}) \in U_\delta \times \tilde V$ and $\tilde V = C_0^{-1} D_0 V = 
\{ C_0^{-1} D_0 \bar{\bs S} \mid \bs S \in \Sym \}$. Moreover, for all $(\bs E, \tilde{\bs S}), (\bs E_1, \tilde{\bs S}_1), (\bs E_2, \tilde{\bs S}_2) \in U_\delta \times \tilde V$, 
\begin{gather}
	\sup_{\substack{\bs E_1 \not = \bs E_2, \tilde{\bs S}, \delta}} \frac{|\tilde{\bs f}_\delta(\bs E_2, \tilde{\bs S}) - \tilde{\bs f}_\delta(\bs E_1, \tilde{\bs S})|}{|\bs E_2 - \bs E_1|} + \sup_{\bs E, \tilde{\bs S}, \delta} \delta^{-1} |\tilde{\bs f}_\delta(\bs E, \tilde{\bs S})| = 1, \label{eq:normstrainlimiting1}\\ 
    \sup_{\bs E, \tilde{\bs S}_1 \not = \tilde{\bs S}_2, \delta} \delta^{-1} \frac{| \tilde{\bs f}_\delta(\bs E, \tilde{\bs S}_2) - \tilde{\bs f}_\delta(\bs E_1, \tilde{\bs S}_1)|}{|\bs S_2 - \bs S_1|} = 1. \label{eq:normstrainlimiting2}
\end{gather}
We say that a family of strain-limiting functions with limiting small strains satisfying \eqref{eq:normstrainlimiting1} and \eqref{eq:normstrainlimiting2} is \emph{normalized}. In this case, all variables representing physical quantities are dimensionless.   
\end{comment}

\subsection{Examples of families of strain-limiting functions}

We now discuss a series of example families of strain-limiting functions with limiting small strains. Let $\bs f: U \times V \rar \Sym$ be a bounded Lipschitz continuous function satisfying,
\begin{gather}
	\forall \bs E, \bar{\bs S}, \,
	|\bs f(\bs E, \bar{\bs S})| \leq \delta_1, \quad \forall \bs E_1 \not = \bs E_2, \bar{\bs S}, \, \frac{|\bs f(\bs E_2, \bar{\bs S}) - \bs f(\bs E_1, \bar{\bs S})|}{|\bs E_2 - \bs E_1|} \leq \tilde C_1, \,\,\,\, \label{eq:fcond1}  \\ 
	\forall \bs E, \bar{\bs S}_1 \not = \bar{\bs S}_2, \quad \frac{| \bs f(\bs E, \bar{\bs S}_2) - \bs f(\bs E, \bar{\bs S}_1)|}{|\bar{\bs S}_2 - \bar{\bs S}_1|} \leq \tilde D_0 \delta_1. \label{eq:fcond2}
\end{gather}
Then
\begin{align}
	\bs f_\delta(\bs E, \bar{\bs S}) := \frac{\delta}{\delta_1} \bs f \Bigl (
	\frac{\delta_1}{\delta} \bs E, \bar{\bs S}
	\Bigr ), \label{eq:generatedfamily}
\end{align}
for $\bs E \in U_\delta := \frac{\delta}{\delta_1}U = \bigl \{ \frac{\delta}{\delta_1} \tilde{\bs E} \mid \tilde{\bs E} \in U \bigr \}$ and $\bar{\bs S} \in V$, defines a family of strain-limiting functions with limiting small strains satisfying $\bs f_{\delta_1} = \bs f$. Moreover, the constants appearing in \eqref{eq:strainlimiting1} and \eqref{eq:strainlimiting2} satisfy
\begin{align}
	C_0 = 1, \quad C_1 = \tilde C_1, \quad D_0 = \tilde D_0. 
\end{align}
Roughly speaking, this example shows that an arbitrary bounded Lipschitz continuous function generates a family of strain-limiting functions with limiting small strains in a natural way. 

We now show that if $\delta_n \rar 0$, then up to passing to a subsequence, a family of strain-limiting functions with limiting small strains and $U_\delta = \delta U$, where $U$ is a fixed open subset in $\Sym$, takes the form \eqref{eq:generatedfamily} to leading order in $\delta_n$. This class of strain-limiting functions with limiting small strains satisfying the additional condition $U_\delta = \delta U$ is extensive and contains all explicit examples discussed in this work. 

\begin{prop}
Let $U, V \subseteq \Sym$ be bounded open sets. For each $\delta > 0$, let
\begin{align}
	U_\delta = \delta U = \{ \delta \tilde{\bs E} \mid \tilde{\bs E} \in U \}. 
\end{align}
Suppose that $\{ \delta_n \}_n$ is a sequence of positive numbers converging to $0$, and $\bs f_\delta : U_\delta \times V \rar \Sym$ is a family of strain-limiting functions with limiting small strains. Then there exists a subsequence $\{\delta_{n_m}\}_m$ and a bounded Lipschitz continuous function $\bs g: U \times V \rar \Sym$ such that 
\begin{align}
	\sup_{(\bs E, \bar{\bs S}) \in U_{\delta_{n_m}} \times V}
	|\bs f_{\delta_{n_m}}(\bs E, \bar{\bs S}) - \delta^{-1}_{n_m}\bs g(\delta^{-1}_{n_m} \bs E, \bar{\bs S})| = o({\delta_{n_m}}),
\end{align}
as $m \rar \infty$. 
\end{prop}

\begin{proof}
Define $\bs h_n : U \times V \rar \Sym$ by 
\begin{align}
	\bs h_n(\tilde{\bs E}, \bar{\bs S}) = \delta_n^{-1} \bs f_{\delta_n}(\delta_n \tilde{\bs E}, \bar{\bs S}), \quad (\tilde{\bs E}, \bar{\bs S}) \in U \times V.
\end{align}
Then by \eqref{eq:strainlimiting1} and \eqref{eq:strainlimiting2}, for all $n$, 
\begin{gather}
	\forall \tilde{\bs E}, \bar{\bs S}, \quad
	|\bs h_n(\tilde{\bs E}, \bar{\bs S})| \leq C_0, \qquad \forall \tilde{\bs E}_1 \not = \tilde{\bs E}_2, \bar{\bs S}, \quad \frac{|\bs h_n(\tilde{\bs E}_2, \bar{\bs S}) - \bs h_n(\tilde{\bs E}_1, \bar{\bs S})|}{|\tilde{\bs E}_2 - \tilde{\bs E}_1|} \leq C_1,  \\ 
	\forall \tilde{\bs E}, \bar{\bs S}_1 \not = \bar{\bs S}_2, \quad \frac{| \bs h_n(\tilde{\bs E}, \bar{\bs S}_2) - \bs h_n(\tilde{\bs E}, \bar{\bs S}_1)|}{|\bar{\bs S}_2 - \bar{\bs S}_1|} \leq D_0.
\end{gather}
By the Arzela-Ascoli theorem, there exists a subsequence $\{ \bs h_{n_m} \}_m$ and $\bs g : U \times V \rar \Sym$ such that 
\begin{align}
	\sup_{(\tilde{\bs E}, \bar{\bs S}) \in U \times V} |\bs h_{n_m}(\tilde{\bs E}, \bar{\bs S}) - \bs g(\tilde{\bs E}, \bar{\bs S})| = o(1), 
\end{align}
as $m \rar \infty$, concluding the proof. 
\end{proof}

A simple subclass of the example \eqref{eq:generatedfamily} is generated by a Lipschitz continuous function $\bs f_1 : V \rar \Sym$ on an open set $V \subseteq \Sym$
satisfying 
\begin{align}
	\forall \bar{\bs S}, \quad |\bs f_1(\bar{\bs S})| \leq 1. 
\end{align}
Then 
\begin{align}
\bs f_\delta(\bar{\bs S}) = \delta \bs f_1(\bar{\bs S}), \label{eq:simplestrainlimiting}
\end{align}
is a family of strain-limiting functions on $V$ with limiting small strains. If there exists a scalar function $W^* : V \rar \bbR$ such that $\bs f_1 = \p_{\bar{\bs S}} W^*$, then the relation
\begin{align}
	\bs E = \delta \bs f_1(\bar{\bs S}) = \p_{\bar{\bs S}}[\delta W^*(\bar{\bs S})], \quad \bar{\bs S} \in V, \label{eq:simplehyper}
\end{align} 
models an elastic solid body with complementary energy $\delta W^*$. If, in addition, $V$ is convex and $W^*$ is a twice continuously differentiable, convex function on $V$, then the relation \eqref{eq:simplehyper} can be inverted yielding 
\begin{align}
	\bar{\bs S} = \p_{\bs E}[\delta W \bigl (\delta^{-1}\bs E\bigr )], \quad \bs E \in \mbox{range}(\delta \bs f_1). \label{eq:Greenelastic}
\end{align} 
In \eqref{eq:Greenelastic}, $W$ is the Legendre transform of $W^*$, and \eqref{eq:Greenelastic} is the classical constitutive relation for a Green elastic solid with strain energy $\delta W(\delta^{-1}\cdot)$ (see \cite{TruesdellNollNLFT}). 

An explicit example of such a family of strain-limiting functions with limiting small strains, inspired by \cite{Raj2010, MAIWALTONELLIP, MAIWALTONMON}, is generated by
\begin{align}
	\bs f_1(\bar{\bs S}) = a ( 1 + a^p |\bar{\bs S}|^{p} )^{-1/p} \bar{\bs S}, \quad \bar{\bs S} \in \Sym, \label{eq:examplestrainlimiting}
\end{align} 
where $a > 0$ and $p \geq 1$. Then for all $\bs E, \bar{\bs S}$, 
\begin{align}
	  |\bs f_1(\bs E, \bar{\bs S})| \leq 1. \label{eq:strainlimitbd1}
\end{align}
Using the product and chain rules, we compute the Frech\'et derivative, 
\begin{align}
	D_{\bar{\bs S}} \bs f_1(\bar{\bs S}) = a (1 + b^p |\bar{\bs S}|^p)^{-1/p} \msb I - 
	 a (1 + b^p |\bar{\bs S}|^p)^{-1/p-1} b^p |\bar{\bs S}|^{p-2} \bar{\bs S} \tens \bar{\bs S},  
\end{align}
where $\msb I : \Sym \rar \Sym$ is the identity map. By the triangle inequality, the operator norm of $D_{\bar{\bs S}} \bs f_1$ satisfies
\begin{align}
	\| D_{\bar{\bs S}} \bs f_1 (\bar{\bs S}) \| \leq a ( 1 + b^p |\bar{\bs S}|^p)^{-1/p} \Bigl (
	1 + b^p (1 + b^p |\bar{\bs S}|^p)^{-1} |\bar{\bs S}|^p \Bigr ) \leq 2 a. \label{eq:Frechetbound1}
\end{align} 
The fundamental theorem of calculus and \eqref{eq:Frechetbound1} then imply that 
\begin{align}
	|\bs f_1 (\bar{\bs S}_2) - \bs f_1 (\bar{\bs S}_1)| &= 
	\Bigl |
	\int_0^1 D_{\bar{\bs S}} \bs f_\delta(s\bar{\bs S}_2 + (1-s)\bar{\bs S}_1)[\bar{\bs S}_2 - \bar{\bs S}_1] ds
	\Bigr | \\&\leq 2\delta a |\bar{\bs S}_2 - \bar{\bs S}_1|. \label{eq:straimlimitbd2}
\end{align}
By \eqref{eq:strainlimitbd1} and \eqref{eq:straimlimitbd2}, we conclude that $\bs f_\delta = \delta \bs f_1$ is a family of strain-limiting functions with limiting small strains. We leave it to the reader to verify that $\bs f_1 = \p_{\bar{\bs S}} W^*$ for an appropriately chosen twice continuously differentiable, convex function $W^*$. 

A family of strain-limiting functions with limiting small strains outside of the scope of our discussion thus far is the following. Let $E_0$ and $a$ be fixed positive constants. For $\delta > 0$ small, define   
\begin{gather}
	 E_{\delta}(\bs E)
	 = \delta^{-1} E_0 \Bigl [
	1 + a \delta^{-1} ([\det(\bs I + 2 \bs E)]^{-1/2} - 1)
	\Bigr ], \\
	\bs f_\delta(\bs E, \bar{\bs S}) = \frac{1+\nu}{E_\delta} \bar{\bs S} - \frac{\nu}{E_\delta} (\tr \bar{\bs S}) \bs I, \label{eq:powerdependentmodulus1}
\end{gather}
One may interpret the constitutive relation $\bs E = \bs f_\delta(\bs E, \bar{\bs S})$ as a generalization of the classical linear constitutive relation for an isotropic solid with a generalized Young's modulus $E_\delta$ depending on the density via $E_\delta = \delta^{-1} E_0 \bigl [ 1 + a \delta^{-1} (\rho/\rho_R - 1) \bigr ]$ (see \eqref{eq:lindensity}). 

Suppose that $b > 0$, $c > 0$ and $ab < 1/2$. We claim that for all $\delta$ sufficiently small (depending on $a$ and $b$), \eqref{eq:powerdependentmodulus1} is a family of strain-limiting functions with limiting small strains on domains 
$B(\bs 0, b \delta) \times B(\bs 0, c)$. Towards proving our claim, we observe that if $\bs E \in B(\bs 0, b \delta)$, then for all $\delta$ sufficiently small, 
\begin{align}
	\Bigl |[\det(\bs I + 2 \bs E)]^{-1/2} - 1 \Bigr | = |-\tr \bs E + o(|\bs E|)| \leq \sqrt{3} |\bs E| + o(|\bs E|) \leq 2 |\bs E|.  
\end{align}
Above we used the fact that by the Cauchy-Schwarz inequality, for any tensor $\bs A$ we have $|\tr \bs A| = |\bs I \cdot \bs A| \leq \sqrt{3} |\bs A|$. By the reverse triangle inequality, we have
\begin{align}
	E_\delta(\bs E) \geq \delta^{-1} E_0 (1 - 2 a \delta^{-1}|\bs E|) \geq \delta^{-1} E_0 (1 - 2 a b) > 0. \label{eq:Edeltalower}
\end{align}
We conclude that
\begin{align}
	 |\bs f_\delta(\bs E, \bar{\bs S}) |
	\leq \frac{1+\nu}{E_\delta}|\bar{\bs S}| + \frac{\nu}{E_\delta}|\tr \bar{\bs S}||\bs I|
	\leq \frac{(1+4\nu)c}{E_0(1 - 2ab)} \delta. \label{eq:densitydep1}
\end{align}
Let $\bs E_1, \bs E_2 \in B(\bs 0, b \delta)$ and $\bar{\bs S}_1, \bar{\bs S}_2 \in B(\bs 0, c)$. Then 
\begin{gather}
\bs f_\delta(\bs E_2, \bar{\bs S}_2) - \bs f_\delta(\bs E_1, \bar{\bs S}_1) =
\frac{E_\delta(\bs E_1) - E_\delta(\bs E_2)}{E_\delta(\bs E_1) E_\delta(\bs E_2)}[(1+\nu)\bar{\bs S}_2 - \nu (\tr \bar{\bs S}_2) \bs I] \\
+ \frac{1}{E_\delta(\bs E_1)}[(1+\nu)(\bar{\bs S}_2 - \bar{\bs S}_1) - \nu (\tr[\bar{\bs S}_2 - \bar{\bs S}_1]) \bs I]. \label{eq:fdeltadiff}
\end{gather}
For all $\delta$ sufficiently small, we have, 
\begin{gather}
\forall s \in [0,1], \quad |(\bs I + s \bs E_1 + (1-s) \bs E_2)^{-1}| \leq (1 - b\delta)^{-1} \leq 2, \\	
\forall s \in [0,1], \quad |\det (\bs I + s \bs E_1 + (1-s) \bs E_2)| \geq 1/2.
\end{gather}
By the fundamental theorem of calculus and the chain rule, we have
\begin{gather}
|E_\delta(\bs E_1) - E_\delta(\bs E_2)| \\= \delta^{-2} E_0 a
\Big | \int_0^1 -\frac{1}{2} [\det (\bs I + s\bs E_1 + (1-s)\bs E_2)]^{-1/2} \\ \qquad \times (\bs I + s\bs E_1 + (1-s)\bs E_2)^{-1} \cdot (\bs E_1 - \bs E_2) ds \Bigr |,
\end{gather}
and thus,
\begin{align}
	|E_\delta(\bs E_1) - E_\delta(\bs E_2)| \leq \delta^{-2} E_0 a \sqrt{2}. \label{eq:Edeltadiff} 
\end{align}
We also have 
\begin{gather}
|(1+\nu)\bar{\bs S}_2 - \nu (\tr \bar{\bs S}_2) \bs I| \leq (1+4\nu)|\bar{\bs S}_2| \leq (1+4\nu)c, \\
|(1+\nu)(\bar{\bs S}_2 - \bar{\bs S}_1) - \nu (\tr[\bar{\bs S}_2 - \bar{\bs S}_1]) \bs I| \leq (1 + 4\nu)|\bar{\bs S}_1 - \bar{\bs S}_2|. 
\end{gather}
Using the previous two estimates along with \eqref{eq:Edeltadiff} and \eqref{eq:Edeltalower}, we conclude from \eqref{eq:fdeltadiff} and repeated use of the triangle inequality that 
\begin{gather}
|\bs f_\delta(\bs E_2, \bar{\bs S}_2) - \bs f_\delta(\bs E_1, \bar{\bs S}_1)| \\ \leq 
 \frac{(1+4\nu)ac \sqrt{2}}{E_0(1-2ab)^2}|\bs E_2 - \bs E_1| + 
 \frac{1+4\nu}{E_0(1-2ab)} \delta |\bar{\bs S}_2 - \bar{\bs S}_1|. \label{eq:densitydep2}
\end{gather}
By \eqref{eq:densitydep1} and \eqref{eq:densitydep2}, we conclude that for all $\delta$ sufficiently small (depending on $a$ and $b$), \eqref{eq:powerdependentmodulus1} is a family of strain-limiting functions with limiting small strains.

A final example of a family of strain-limiting functions with limiting small strains is given in terms of a slightly different form of a density dependent generalized Young's modulus: 
\begin{gather}
	E_\delta = \delta^{-1} E_0 \Bigl [
	1 + a \delta^{-1} ([\det (\bs I + 2 \bs E)]^{1/2} - 1)
	\Bigr ]^{-1}, \\
		\bs f_\delta(\bs E, \bar{\bs S}) = \frac{1+\nu}{E_\delta} \bar{\bs S} - \frac{\nu}{E_\delta} (\tr \bar{\bs S}) \bs I, \label{eq:powerdependentmodulus2}
\end{gather}
Again, by \eqref{eq:lindensity}, one may interpret the generalized Young's modulus as depending on the density via $E_\delta = \delta^{-1} E_0 \bigl [ 1 + a \delta^{-1} (\rho_R/\rho - 1) \bigr ]^{-1}$.
We note that both \eqref{eq:powerdependentmodulus1} and \eqref{eq:powerdependentmodulus2} have the physical property that an increase in density, $\rho$, yields an increase in the Young's modulus, $E_\delta$. Arguing as above, one may conclude that \eqref{eq:powerdependentmodulus2} also yields a family of strain-limiting functions with limiting small strains on domains $U_\delta \times V = B(\bs 0, b) \times B(\bs 0, c)$ for appropriate choices of $b$ and $c$.

\subsection{Nonlinear relations between stress and linearized strain} 

We now show that a family of constitutive relations of the form $\bs E = \bs f_\delta(\bs E, \bar{\bs S})$, with $\bs E$ interpreted as the Green-Saint Venant strain tensor associated to a deformation and $\bar{\bs S}$ as the symmetric Piola-Kirchhoff stress, asymptotically yield nonlinear relations between stress and the associated linearized strain $\bs \eps$, up to a $\bs O(\delta^2)$ error. Moreover, Corollary \ref{c:1} and the examples discussed thereafter show that in general, the resulting asymptotic relations between stress and linearized strain can be genuinely nonlinear to leading order in $\delta$. This stands in contrast to the setting of a fixed constitutive relation, which under mild differentiability assumptions, always can be asymptotically reduced to a linear relation between stress and linearized strain to leading order in the displacement gradient $|\bs F - \bs I|$. 

The preliminary set-up for our result is as follows. For an element $\bs E_\delta \in B(\bs 0, 1/2)$, rotation $\bs R_\delta$, and $\bar{\bs S} \in \Sym$, we associate a deformation gradient $\bs F_\delta$ via
\begin{gather}
	\bs C_\delta := \bs I +2\bs E_\delta, \quad \bs F_\delta := \bs R_\delta \bs C_\delta^{1/2}, \quad  \bs \eps_\delta := \frac{1}{2} (\bs F_\delta + \bs F_\delta^T) - \bs I, \label{eq:Fdeltaeq}
\end{gather}
and a stress variable, $\bs \sigma_\delta$, by
\begin{align}
	\bs \sigma_\delta := \frac{1}{2} \bigl (\bs F_\delta \bar{\bs S} + \bar{\bs S} \bs F_\delta^{T} \bigr ). 
\end{align}
The tensor $\bs \sigma_\delta$ can be interpreted as the symmetric part of an associated first Piola-Kirchhoff stress tensor.  

\begin{prop}\label{p:2}
	Let $\bs f_\delta : U_\delta \times V \rar \Sym$ be a family of strain-limiting functions with limiting small strains. Let $\bar{\bs S} \in V$. Assume that there exists $r > 0$ such that for each $\delta > 0$ sufficiently small, there exists $\bs E_\delta \in U_\delta$ such that $B(\bs E_\delta, r \delta) \subseteq U_\delta$ and 
	\begin{gather}
		%	\bs E = \frac{1}{2} \log (I + \nabla \bs u + (\nabla \bs u)^T + \nabla \bs u (\nabla \bs u)^T), \\ 
		\bs E_\delta = \bs f_\delta(\bs E_\delta, \bar{\bs S}). \label{eq:Hdeltaeq2}
	\end{gather}
	In addition, suppose that there exists a fixed dimensionless constant $C_2 > 0$ such that for all $\delta$ sufficiently small,
	\begin{align}	
|\bs R_\delta - \bs I| < C_2 \delta. \label{eq:smallrotation}
	\end{align}

	Then for all $\delta$ sufficiently small, $\bs (\bs \eps_\delta, \bs \sigma_\delta) \in U_\delta \times V$, and as $\delta \rar 0$, 
	\begin{align}
		\bs \eps_\delta = \bs f_\delta(\bs \eps_\delta, \bs \sigma_\delta) + \bs O(\delta^2), \label{eq:smallstrain3}
	\end{align}
	where the big-oh term depends on $C_0$, $C_1$, $C_2$, and $D_0|\bar{\bs S}|$. 
\end{prop}

\begin{proof}
	For all $\delta > 0$ sufficiently small, we use the bound $|(\bs I + 2 \bs E_\delta)^{1/2}| \leq (3 + 2\sqrt{3}|\bs E_\delta|)^{1/2}$, \eqref{eq:smallrotation}, \eqref{eq:Hdeltaeq2} and \eqref{eq:strainlimiting1} to obtain
	\begin{align}
		|\bs F_\delta - \bs I| &= |\bs R_\delta (\bs I + 2\bs E_\delta)^{1/2} - \bs I| \\
		&\leq  |(\bs R_\delta - \bs I) (\bs I + 2\bs E_\delta)^{1/2}| + |(\bs I + 2\bs E_\delta)^{1/2} - \bs I| \\
		&\leq (3 + 2 \sqrt{3} |\bs E_\delta|)^{1/2} |\bs R_\delta - \bs I| + 
		(1 - 2 C_0 \delta)^{-1/2}|\bs E_\delta| \\
		&\leq \Bigl [(3 + 2 \sqrt{3} C_0 \delta)^{1/2} C_2 + (1 - 2 C_0 \delta)^{-1/2} C_0 \Bigr ] \delta. \label{eq:Fdeltaest}
	\end{align}
	Since $\bs C_\delta = \bs F_\delta^T \bs F_\delta$ and $\bs E_\delta = \frac{1}{2} (\bs C_\delta - \bs I)$, we conclude that
	\begin{align}
		\bs C_\delta = \bs I + 2\bs \eps_\delta + \bs O(\delta^2), \quad \bs E_\delta = \bs \eps_\delta + \bs O(\delta^2), \label{eq:linstrains3}
	\end{align}
	and 
	\begin{align}
		\bs \sigma_\delta = \frac{1}{2}(\bs I + \bs O(\delta))\bar{\bs S} + \frac{1}{2}\bar{\bs S}(\bs I + \bs O(\delta)^T) = 
		\bar{\bs S} + \bs O(\delta)\bar{\bs S} + \bar{\bs S}\bs O(\delta)^T, \label{eq:linstress}
	\end{align}
	where the big-oh terms depend only on $C_0$ and $C_2$.
	
	Since $\bar{\bs S} \in V$ and $V$ is open, there exists $r_0 > 0$ such that $B(\bar{\bs S}, r_0) \subseteq V$. The relations \eqref{eq:linstrains3} and \eqref{eq:linstress} imply that there exists a constant $C$ depending on $C_0$ and {$C_2$} such that for all $\delta$ sufficiently small, 
	\begin{gather}
		|\bs E_\delta - \bs \eps_\delta| \leq C \delta^2 < r \delta, \\
		|\bs \sigma_\delta - \bar{\bs S}| \leq C \delta |\bar{\bs S}|, \label{eq:estimatestress}
	\end{gather} 
	and thus, $(\bs \eps_\delta, \bar{\bs S}) \in B(\bs E_\delta, r \delta) \times B(\bar{\bs S}, r_0) \subseteq U_\delta \times V$ for all $\delta$ sufficiently small. Then by \eqref{eq:Hdeltaeq2}, \eqref{eq:linstrains3}, \eqref{eq:strainlimiting1}, and \eqref{eq:strainlimiting2}, 
	\begin{align}
		\bs \eps_\delta &= \bs f_\delta(\bs E_\delta, \bar{\bs S}) + \bs \eps_\delta - \bs E_\delta \\
		&= \bs f_\delta(\bs \eps_\delta, \bs \sigma_\delta) + [\bs f_\delta(\bs \eps_\delta, \bar{\bs S}) - \bs f_\delta(\bs \eps_\delta, \bs \sigma_\delta) + \bs \eps_\delta - \bs E_\delta + \bs f_\delta(\bs E_\delta, \bar{\bs S}) - \bs f_\delta(\bs \eps_\delta, \bar{\bs S})] \\
		&= \bs f_\delta(\bs \eps_\delta, \bs \sigma_\delta) + \bs O(\delta^2), 
	\end{align}
	as $\delta \rar 0$, where the big-oh term depends only on $C_0$, $C_1$, $C_2$, and $|D_0||\bar{\bs S}|$. This concludes the proof. 
\end{proof}

We remark that the approximation \eqref{eq:linstress} between the stress variables and the associated bound \eqref{eq:estimatestress} have the same form as
the classical setting where the norm of the displacement gradient $\delta_0$ serves as the asymptotic parameter rather than $\delta$. Indeed, consider the one-dimensional example from Section 1.2 and the notation therein. Let $\sigma_\delta = \eps/(a \delta)$, and let $\delta$, $a$, and $p$ be fixed (i.e., the constitutive relation is fixed). From \eqref{eq:delta0approximation} and the fact that $\delta_0 = |\eps|$, we conclude that 
\begin{align}
\sigma_\delta = \bar S\bigl [1 + O(\delta_0/\delta)\bigr ]^{-1} = \bar S \bigl [1 + O(\delta_0/\delta)\bigr ], \label{eq:onedestimate}
\end{align}
for all $\delta_0$ sufficiently small. We then obtain 
\begin{align}
	|\sigma_\delta - \bar S| \leq \frac{C}{ \delta} |\bar S| \delta_0, \label{eq:estimateonedstress}
\end{align}
for some absolute constant $C > 0$. This is the same form as \eqref{eq:estimatestress} where $\delta_0$ is now the small changing parameter rather than $\delta$. 

Moreover, we comment that inserting $a \bar S = O(\delta_0/\delta)$ (derived from \eqref{eq:delta0approximation}) into the right-hand side of \eqref{eq:estimateonedstress} gives the \textit{false} impression that using $\delta_0$ as the asymptotic parameter results in a ``stronger" approximation between the stress variables
\begin{align}
	|\sigma_\delta - \bar S| \leq \frac{\tilde C}{a \delta^2} \delta_0^2. \label{eq:quadratic} 
\end{align}
One should not necessarily view the above estimate as stronger than \eqref{eq:estimatestress} since $a \sigma_\delta$ and $a \bar S$ are $O(\delta_0/\delta)$, but in \eqref{eq:estimatestress}, $\bs \sigma_\delta$ and $\bar{\bs S}$ are $O(1)$ relative to a fixed stress scale.

\begin{cor}\label{c:1}
	In addition to the assumptions of Proposition \ref{p:2}, suppose that there exist an open set $U \subseteq \Sym$, $\bs f_1 : U \times  V \rar \Sym$, and a constant $C_3 > 0$ such that 
	\begin{align}
		U_\delta = \delta U = \{ \delta \tilde{\bs E} \mid \tilde{\bs E} \in U \}, 
	\end{align}
	and 
	\begin{align}
		\forall (\bs E, \bar{\bs S}) \in U_\delta \times V, \quad 
		\bigl |
		\bs f_\delta(\bs E, \bar{\bs S}) - \delta \bs f_1 \bigl ( \delta^{-1} \bs E, \bar{\bs S}) 
		\bigr | \leq C_3 \delta^2. \label{eq:addassumpt}
	\end{align}
	
	Then, in the notation of Proposition \ref{p:2}, we have as $\delta \rar 0$,  
	\begin{align}
		\bs \eps_\delta = \delta \bs f_1(\delta^{-1} \bs \eps_\delta, \bs \sigma_\delta) + \bs O(\delta^2), 
	\end{align}
	where the big-oh term depends on $C_0, C_1, C_2$, $C_3$, and $D_0|\bar{\bs S}|$. 
\end{cor}

We now discuss the examples from Section 3.2 within the context of Proposition \ref{p:2} and Corollary \ref{c:1}. It is clear that a general family of strain-limiting functions with limiting small strains given by \eqref{eq:simplestrainlimiting} satisfies the hypotheses of both Proposition \ref{p:2} and Corollary \ref{c:1}. In particular, if \eqref{eq:simplehyper} holds with a twice continuously differentiable, convex function $W^*$, then Corollary \ref{c:1} implies in the notation therein that
\begin{align}
		\bs \eps_\delta =  \p_{\bs \sigma_\delta}[\delta W^*(\bs \sigma_\delta)] + \bs O(\delta^2), \quad \bs \sigma_\delta \in V. \label{eq:simplehyperlin}
\end{align}
Omitting the $\bs O(\delta^2)$ term yields $\bs \eps_\delta =  \p_{\bs \sigma_\delta}[\delta W^*(\bs \sigma_\delta)]$ or, equivalently, 
\begin{align}
\bs \sigma_\delta = \p_{\bs \eps_\delta}[\delta W \bigl (\delta^{-1}\bs \eps_\delta \bigr )], \label{eq:Greenelasticlin}
\end{align}
where $W$ is the Legendre transform of $W^*$.
The relation \eqref{eq:Greenelasticlin} models a Green elastic solid, i.e., the stress is derived from a strain energy depending on the linearized strain. In contrast to classical linearized elasticity, however, our framework allows for strain energies $W$ beyond quadratic forms of the linearized strain.  

We now consider the explicit family \eqref{eq:powerdependentmodulus1}. We write $\bs f_\delta(\bs E, \bar{\bs S}) = \delta \bs g_\delta(\bs E/\delta, \bar{\bs S})$ where 
\begin{align}
	\bs g_\delta(\tilde{\bs E}, \bar{\bs S}) = \frac{1+\nu}{\delta E_\delta(\delta \tilde{\bs E})} \bar{\bs S} - \frac{\nu}{\delta E_\delta(\delta \tilde{\bs E})} (\tr \bar{\bs S}) \bs I, \quad (\tilde{\bs E}, \bar{\bs S}) \in B(\bs 0, b) \times B(\bs 0, c),
\end{align}
and we note that for all $\tilde \delta$ and $b$ sufficiently small,
\begin{align}
	(0,\tilde \delta) \times B(\bs 0, b) \ni (\delta, \tilde{\bs E}) \mapsto \delta E_\delta(\delta \tilde{\bs E}) &= E_0 \Bigl [
	1 + a \delta^{-1} ([\det(\bs I + 2 \delta \tilde{\bs E})]^{-1/2} - 1)
	\Bigr ],
\end{align} extends smoothly to a function of $(\delta, \tilde{\bs E}) \in (-\tilde \delta, \tilde \delta) \times B(\bs 0, b)$. Moreover, we have for all $\tilde{\bs E} \in B(\bs 0, b)$, 
\begin{align}
	\delta E_\delta(\delta \tilde{\bs E}) = E_0 [
	1 - a \tr \tilde{\bs E} + O(a b^2 \delta)
	]. \label{eq:Edeltasmalldelta}
\end{align}
We observe that $\bs E = \bs f_\delta(\bs E, \bar{\bs S})$ if and only if 
\begin{align}
	\tilde{\bs E} = \bs g_\delta(\tilde{\bs E}, \bar{\bs S}), \label{eq:tilderelation}
\end{align}
where $\bs E = \delta\cdot \tilde{\bs E}$, $\bs g_0(\bs 0, \bs 0) = \bs 0$, and 
\begin{align}
D_{\tilde{\bs E}} \bigl (
\tilde{\bs E} - \bs g_\delta(\tilde{\bs E}, \bar{\bs S})
\bigr ) \Big |_{(\delta, \tilde{\bs E}, \bar{\bs S}) = (0,\bs 0, \bs 0)} = \msb I.
\end{align} 
Thus, by the implicit function theorem and choosing $\tilde \delta$, $b$, and $c$ sufficiently small, we conclude that for each $0 < \delta < \tilde \delta$ and $\bar{\bs S} \in B(\bs 0, c)$, there exists a unique $\tilde{\bs E}_\delta \in B(\bs 0, b)$ satisfying \eqref{eq:tilderelation}. Written differently, we have shown that for each $0 < \delta < \tilde \delta$ and $\bar{\bs S} \in B(\bs 0, c)$, there exists a unique $\bs E_\delta \in B(\bs 0, \delta b)$ satisfying $\bs E_\delta = \bs f_\delta(\bs E_\delta, \bar{\bs S})$, and the hypotheses of Proposition \ref{p:2} are satisfied on $U_\delta \times V = B(\bs 0, b \delta) \times B(\bs 0, c)$. 

Finally, define 
\begin{align}
	\bs f_1(\tilde{\bs E}, \bar{\bs S}) = \bs g_0(\tilde{\bs E}, \bar{\bs S})
	= E_0^{-1} (
	1 - a \tr \tilde{\bs E}
	)^{-1}\bigl [
	(1+\nu) \bar{\bs S} - \nu (\tr \bar{\bs S}) \bs I
	\bigr ]. 
\end{align}
By \eqref{eq:Edeltasmalldelta}, it follows that for all $\bs E \in B(\bs 0, \delta b)$ (with $ab < 1/2$), 
\begin{align}
	\Bigl | \frac{1}{E_\delta(\bs E)} - \frac{\delta}{E_0(1 - a\tr (\delta^{-1} \bs E))} \Bigr | &= 
 \frac{\bigl | \delta^{-1} E_0(1 - a \tr(\delta^{-1}\bs E)) - E_\delta(\bs E) \bigr |}{
	\delta^{-1} E_0(1 - a\tr (\delta^{-1} \bs E)) E_\delta(\bs E)} \\
	&\leq \frac{a \bigl |1 - [\det(\bs I + 2 \bs E)]^{-1/2} - \tr \bs E \bigr |}{E_0(1 -2 ab)^2} \\
	&\leq \frac{C a |\bs E|^2}{E_0(1 - 2ab)^2} \\
	&\leq \frac{C a b^2 \delta^2}{E_0(1 - 2ab)^2},
\end{align}
where $C$ is an absolute constant. Thus, for all $(\bs E, \bar{\bs S}) \in B(\bs 0, b \delta) \times B(\bs 0, c)$,  
\begin{align}
	\bigl |
	\bs f_\delta(\bs E, \bar{\bs S}) - \delta \bs f_1 \bigl ( \delta^{-1} \bs E, \bar{\bs S}) 
	\bigr | \leq \frac{C ab^2 (1 + 4\nu) c \delta^2}{E_0(1 - 2ab)^2}. \label{eq:estimate}
\end{align} 
The estimate \eqref{eq:estimate} proves that the hypotheses of Corollary \ref{c:1} are satisfied. In the notation therein,  
\begin{align}
	\bs \eps_\delta = \delta E_0^{-1} \bigl [
	1 - a \delta^{-1} \tr(\bs \eps_\delta)
	\bigr ]^{-1}\bigl [
	(1+\nu) \bs \sigma_\delta - \nu (\tr \bs \sigma_\delta) \bs I
	\bigr ] + \bs O(\delta^2), \label{eq:lindensitydependent1}
\end{align}
as $\delta \rar 0$, and omitting the $\bs O(\delta^2)$ term yields  
\begin{align}
	\bs \eps_\delta = \delta E_0^{-1} \bigl [
	1 - a \delta^{-1} \tr(\bs \eps_\delta)
	\bigr ]^{-1}\bigl [
	(1+\nu) \bs \sigma_\delta - \nu (\tr \bs \sigma_\delta) \bs I
	\bigr ]. \label{eq:lindensitydependent1b}
\end{align}

One may similarly show that the hypotheses of Proposition \ref{p:2} and Corollary \ref{c:1} are satisfied for \eqref{eq:powerdependentmodulus2} on a family of nontrivial open sets $U_\delta \times V$ with 
\begin{align}
	\bs f_1(\tilde{\bs E}, \bar{\bs S}) = E_0^{-1} (
	1 + a \tr \tilde{\bs E}
	)\bigl [
	(1+\nu) \bar{\bs S} - \nu (\tr \bar{\bs S}) \bs I
	\bigr ].
\end{align}
In the notation of Proposition \ref{p:2} and Corollary \ref{c:1}, we conclude that
\begin{align}
	\bs \eps_\delta = \delta E_0^{-1} \bigl [
	1 + a \delta^{-1} \tr(\bs \eps_\delta)
	\bigr ]\bigl [
	(1+\nu) \bs \sigma_\delta - \nu (\tr \bs \sigma_\delta) \bs I
	\bigr ] + \bs O(\delta^2), \label{eq:lindensitydependent2}
\end{align}
as $\delta \rar 0$, and omitting the $\bs O(\delta^2)$ term yields
\begin{align}
	\bs \eps_\delta = \delta E_0^{-1} \bigl [
	1 + a \delta^{-1} \tr(\bs \eps_\delta)
	\bigr ]\bigl [
	(1+\nu) \bs \sigma_\delta - \nu (\tr \bs \sigma_\delta) \bs I
	\bigr ]. \label{eq:lindensitydependent2b}
\end{align}

Constitutive relations of the form 
\begin{align}
	\bs \eps = E_{\ms{ref}}^{-1}\bigl [ 1 - \Gamma \tr \bs \eps \bigr ]^{-1} \bigl [
	(1+\nu) \bs \sigma - \nu (\tr \bs \sigma) \bs I
	\bigr ], \label{eq:lindensityconst1}
\end{align}
or 
\begin{align}
	\bs \eps = E_{\ms{ref}}^{-1}\bigl [ 1 + \Gamma \tr \bs \eps \bigr ] \bigl [
	(1+\nu) \bs \sigma - \nu (\tr \bs \sigma) \bs I
	\bigr ], \label{eq:lindensityconst2}
\end{align}
where $\bs \sigma$ is the stress variable, have been appeared extensively in the literature; see e.g. \cite{RAJDENSITY21, D32021, D32021II, MURRURAJ21UNI, MURRAJRAJ21BI, ITOUKOVRAJ21, PRUSARAJWINE22, VAJMURRAJ22}. Our framework and the expansions \eqref{eq:lindensitydependent1} and \eqref{eq:lindensitydependent2} place \eqref{eq:lindensityconst1} and \eqref{eq:lindensityconst2} on firm mathematical footing as the asymptotic leading order relations (in $\delta$) arising from a constitutive theory for finite deformations. 

We conclude this section by noting that our results are not restricted solely to constitutive relations expressed via the Green-Saint Venant strain tensor and the symmetric Piola-Kirchhoff stress tensor. We demonstrate that analogues of Proposition \ref{p:2} and Corollary \ref{c:1} can be proved in terms of variables $\bs H$ and $\bs T$ representing the Hencky strain and Cauchy stress {for an isotropic elastic solid}. The preliminary set-up is as follows. For an element $(\bs H_\delta, \bs T) \in \Sym \times \Sym$ and rotation $\bs R_\delta$, we associate a deformation gradient $\bs F_\delta$ via
\begin{gather}
	\bs F_\delta := e^{\bs H_\delta}\bs R_\delta, \quad \bs \eps_\delta := \frac{1}{2} (\bs F_\delta + \bs F_\delta^T) - \bs I, 
\end{gather}
and a stress variable, $\bs \sigma_\delta$, via 
\begin{align}
\bs \sigma_\delta = (\det \bs F_\delta) \frac{1}{2} \bigl ( \bs T \bs F_\delta^{-T} + \bs F_\delta^{-1} \bs T \bigr ).
\end{align}
As before, the tensor $\bs \sigma_\delta$ can be interpreted as the symmetric part of an associated first Piola-Kirchhoff stress tensor.

\begin{prop}\label{p:3}
Let $\bs f_\delta : U_\delta \times V \rar \Sym$ be a family of strain-limiting functions with limiting small strains. Let $\bs T \in V$. Assume that there exists $r > 0$ such that for each $\delta > 0$ sufficiently small, there exists $\bs H_\delta \in U_\delta$ such that $B(\bs H_\delta,r\delta) \subseteq U_\delta$ and  
\begin{gather}
	%	\bs H = \frac{1}{2} \log (I + \nabla \bs u + (\nabla \bs u)^T + \nabla \bs u (\nabla \bs u)^T), \\ 
	\bs H_\delta = \bs f_\delta(\bs H_\delta, \bs T). \label{eq:Hdeltaeq}
\end{gather}
In addition, suppose that there exists a fixed dimensionless constant $C_2 > 0$ such that for all $\delta$ sufficiently small,
\begin{align}	
|\bs R_\delta - \bs I| < C_2 \delta. \label{eq:smallrotation2}
\end{align}

Then for all $\delta$ sufficiently small, $\bs (\eps_\delta, \bs \sigma_\delta) \in U_\delta \times V$, and as $\delta \rar 0$, 
\begin{align}
	\bs \eps_\delta = \bs f_\delta(\bs \eps_\delta, \bs \sigma_\delta) + \bs O(\delta^2), \label{eq:smallstrain4}
\end{align}
where the big-oh term depends only on $C_0$, $C_1$, {$C_2$}, and $D_0|\bs T|$.
\end{prop}

\begin{proof}
For all $\delta > 0$ sufficiently small, we have by {\eqref{eq:smallrotation2}}, \eqref{eq:Hdeltaeq} and \eqref{eq:strainlimiting1} that
\begin{align}
	|\bs F_\delta - \bs I| &= |e^{\bs H_\delta} \bs R_\delta - \bs I| \\
	&\leq  |e^{\bs H_\delta} (\bs R_\delta - \bs I)| + |e^{\bs H_\delta} - \bs I| \\
	&\leq  e^{|\bs H_\delta|} |\bs R_\delta - \bs I| + 
|\bs H_\delta| e^{|\bs H_\delta|} \\
&\leq [C_2 + C_0] e^{C_0 \delta}  \delta.
\end{align}
Since $\bs H_\delta = \frac{1}{2} \log \bs F_\delta \bs F^T_\delta$, we conclude that
\begin{align}
\bs H_\delta = \bs \eps_\delta + \bs O(\delta^2), \quad
\bs \sigma_\delta = \frac{1}{2}\bs T(\bs I + \bs O(\delta) + \frac{1}{2}(\bs I + \bs O(\delta))\bs T,
\end{align}
where the big-oh terms depend only on $C_0$, $C_1$, {and $C_2$}. 

Arguing as in Proposition \ref{p:2}, we conclude that for all $\delta$ sufficiently small, $(\bs \eps_\delta, \bs \sigma_\delta) \in U_\delta \times V$, and
\begin{align}
	\bs \eps_\delta &= \bs f_\delta(\bs H_\delta, \bs T) + \bs \eps_\delta - \bs H_\delta \\
	&= \bs f_\delta(\bs \eps_\delta, \bs \sigma_\delta) + [ \bs f_\delta(\bs \eps_\delta, \bs T)-\bs f_\delta(\bs \eps_\delta, \bs \sigma_\delta)+\bs \eps_\delta - \bs H_\delta + \bs f_\delta(\bs H_\delta, \bs T) - \bs f_\delta(\bs \eps_\delta, \bs T)] \\
	&= \bs f_\delta(\bs \eps_\delta, \bs \sigma_\delta) + \bs O(\delta^2), 
\end{align}
as $\delta \rar 0$, where the big-oh term depends only on $C_0$, $C_1$, {$C_2$}, and $D_0|\bs T|$. This concludes the proof. 
\end{proof}

\section{Conclusion}
In this study, we have presented a rigorous asymptotic framework that provides a mathematical foundation for nonlinear constitutive relations between stress and linearized strain. Within our framework, we have demonstrated that a nonlinear constitutive relation between stress and linearized strain emerges as the primary term of an asymptotic expansion in $\delta$ of a family of strain-limiting, nonlinear constitutive relations,
\begin{align}
	\bs E = \bs f_\delta(\bs E, \bar{\bs S})
\end{align}
between stress and nonlinear strain. Here, the dimensionless parameter $\delta$ determining the leading order term is (up to a fixed constant) the limiting small strain of $\bs f_\delta$: $\forall \bs E, \bar{\bs S},$ $|\bs f_\delta(\bs E, \bar{\bs S})| \leq \delta$. {Although the size of the strain $\bs E$ is controlled by $\delta$ (and constrained to be small) via the associated constitutive relation, the size of the stress $\bar{\bs S}$, relative to a fixed stress scale, is neither required to be controlled by $\delta$ nor required to be small.} Our approach diverges from classical linearized elasticity where a constitutive relation is fixed and the asymptotic parameter is the size of the displacement gradient $\delta_0 = |\bs F - \bs I|$. 

As our results show, strain-limiting constitutive relations $\bs E = \bs f_\delta(\bs E, \bar{\bs S})$ are approximated by $\bs \eps = \bs f_\delta(\bs \eps, \bs \sigma)$ up to a quadratic error in $\delta$ with $\bs \sigma$ interpreted as the symmetric part of the associated first Piola-Kirchhoff stress tensor; however, a fundamental question remains. In particular, we conjecture that for a fixed external body force $\bs b$ and for all $\delta$ sufficiently small, solvability of the ``linearized" equilibrium equations, 
\begin{gather}
	\bs 0 = \Div \bs \sigma_{L,\delta} + \bs b, \quad \bs \sigma_{L,\delta}^T = \bs \sigma_{L,\delta}, \\
	\bs \eps_{L,\delta} = \bs f_\delta(\bs \eps_{L,\delta}, \bs \sigma_{L,\delta}), 
\end{gather}
implies solvability of the fully nonlinear equilibrium equations,
\begin{gather}
	\bs 0 = \mbox{Div}\, \bs S_\delta + \bs b, \quad \bs S_\delta \bs F_\delta^T = \bs F_\delta \bs S_\delta^T, \\
	\bs E_\delta = \bs f_\delta(\bs E_\delta, \bar{\bs S}_\delta), 
\end{gather}
and
\begin{gather}
\bs E_\delta = \bs \eps_{L,\delta} + \bs O(\delta^2), \quad
\bar{\bs S}_\delta = \bs \sigma_{L,\delta} + \bs O(\delta), \quad \bs S_\delta = \bs \sigma_{L,\delta} + \bs O(\delta),
\end{gather}
as $\delta \rar 0$.\footnote{Of course, on a bounded domain $\cl B$, a fixed set of boundary conditions must also be given.} Such a fundamental result would be unique but in line with analogous results established for classical linearized elasticity; see, e.g., the classic works of Stoppelli \cite{Stoppelli1954, Stoppelli1955} and the discussion of more recent results by Ciarlet in Chapter 6 of \cite{Ciarlet00Book1}.

\bibliographystyle{plain}
\bibliography{researchbibmech}

\begin{thebibliography}{10}

\bibitem{MechanicsSolidsI}
J.~F. Bell.
\newblock {\em The {E}xperimental {F}oundations of {S}olid {M}echanics}.
\newblock Mechanics of Solids, I. Springer-Verlag, Berlin, 1984.
\newblock Reprint of the 1973 original.

\bibitem{Bullfinger1729}
G.~B. Bullfinger.
\newblock De solidorum resistentia specimen.
\newblock {\em Commentari Accademiae Scientiarum}, 4:140--155, 1729.

\bibitem{Ciarlet00Book1}
P.~G. Ciarlet.
\newblock {\em Mathematical elasticity. {V}olume {I}. {T}hree-dimensional
  elasticity}, volume~84 of {\em Classics in Applied Mathematics}.
\newblock Society for Industrial and Applied Mathematics (SIAM), Philadelphia,
  PA, [2022] \copyright 2022.
\newblock Reprint of the 1988 edition [0936420].

\bibitem{Grasleyetal2015}
Z.~Grasley, R.~El-Helou, M.~D’Ambrosia, D.~Mokarem, C.~Moen, and
  K.~Rajagopal.
\newblock Model of infinitesimal nonlinear elastic response of concrete
  subjected to uniaxial compression.
\newblock {\em Journal of Engineering Mechanics}, 141(7):04015008, 2015.

\bibitem{HAO05}
Y.~L. Hao, S.~J. Li, S.~Y. Sun, C.~Y. Zheng, Q.~M. Hu, and R.~Yang.
\newblock Super-elastic titanium alloys with unstable plastic deformation.
\newblock {\em Appl. Phy. Lett.}, 87, 2005.

\bibitem{ITOUKOVRAJ21}
H.~Itou, V.~Kovtunenko, and K.~R. Rajagopal.
\newblock On an implicit model linear in both stress and strain to describe the
  response of porous solids.
\newblock {\em Math. Mech. Solids}, 144:107--118, 2021.

\bibitem{RajetalBones2024}
A.~Jeyavel, P.~Alagappan, J.~Bird, M.~Moreno, and K.~R. Rajagopal.
\newblock A new constitutive relation to describe the response of bones.
\newblock {\em Int. J. Non-Linear Mech.}, 61:104664, 2024.

\bibitem{LI2007}
T.~Li, J.~W. Morris, N.~Nagasako, S.~Kuramoto, and D.~C. Chrzan.
\newblock ``ideal'' engineering alloys.
\newblock {\em Phys. Rev. Lett.}, 98:105503, Mar 2007.

\bibitem{MAIWALTONELLIP}
T.~Mai and J.~Walton.
\newblock On strong ellipticity for implicit and strain-limiting theories of
  elasticity.
\newblock {\em Math. Mech. Solids}, 20:121--139, 2015.

\bibitem{MAIWALTONMON}
T.~Mai and J.~Walton.
\newblock On strong monotonicity for strain-limiting theories of elasticity.
\newblock {\em J. Elast.}, 120:39--65, 2015.

\bibitem{Morganetal2001}
E.~R. Morgan, O.~C. Yeh, W.~C. Chang, and T.~M. Keaveny.
\newblock Nonlinear behavior of trabecular bone at small strains.
\newblock {\em J. Biomech. Eng.}, 123:1--9, 2001.

\bibitem{Moyer75}
B.~Moyer.
\newblock Robert {H}ooke's ambiguous presentation of `{H}ooke's {L}aw'.
\newblock {\em Isis}, 68:275--288, 1975.

\bibitem{MURRAJRAJ21BI}
P.~Murru and K.~R. Rajagopal.
\newblock Stress concentration due to the bi-axial deformation of a plate with
  a porous elastic body with a hole.
\newblock {\em Z. Angew. Math. Mech.}, 2021.

\bibitem{MURRURAJ21UNI}
P.~Murru and K.~R. Rajagopal.
\newblock Stress concentration due to the presence of a hole within the context
  of elastic bodies.
\newblock {\em Material Design \& Processing Communications}, 3(5):e219, 2021.

\bibitem{D32021}
P.~Murru, C.~Torrence, Z.~Grasley, K.~R. Rajagopal, P.~Alagappan, and
  E.~Garboczi.
\newblock Density-driven damage mechanics ({D3-M}) model for concrete {I}:
  mechanical damage.
\newblock {\em Int. J. Pavement Eng.}, 23(4):1161--1174, 2022.

\bibitem{D32021II}
P.~Murru, C.~Torrence, Z.~Grasley, K.~R. Rajagopal, P.~Alagappan, and
  E.~Garboczi.
\newblock Density driven damage mechanics ({D3-M}) model for concrete {II}:
  fully coupled chemo-mechanical damage.
\newblock {\em Int. J. Pavement Eng}, 23(4):1175--1185, 2022.

\bibitem{PRUSARAJWINE22}
V.~Prusa, K.~R. Rajagopal, and A.~Wineman.
\newblock Pure bending of an elastic prismatic beam made of a material with
  density-dependent material parameters.
\newblock {\em Math. Mech. Solids}, 27:1546--1558, 2022.

\bibitem{PRUSA2020}
V.~Průša, K.R. Rajagopal, and K.~Tůma.
\newblock Gibbs free energy based representation formula within the context of
  implicit constitutive relations for elastic solids.
\newblock {\em Int. J. Non-Linear Mech.}, 121:103433, 2020.

\bibitem{Raj_Implicit03}
K.~R. Rajagopal.
\newblock On implicit constitutive theories.
\newblock {\em Appl. Math.}, 48(4):279--319, 2003.

\bibitem{RajElastElast}
K.~R. Rajagopal.
\newblock The elasticity of elasticity.
\newblock {\em Z. Angew. Math. Phys.}, 58(2):309--317, 2007.

\bibitem{Raj2010}
K.~R. Rajagopal.
\newblock On a new class of models in elasticity.
\newblock {\em Math. Comput. Appl.}, 15(4):506--528, 2010.

\bibitem{RajConspectus}
K.~R. Rajagopal.
\newblock Conspectus of concepts of elasticity.
\newblock {\em Math. Mech. Solids}, 16(5):536--562, 2011.

\bibitem{RajSmallStrain}
K.~R. Rajagopal.
\newblock Non-linear elastic bodies exhibiting limiting small strain.
\newblock {\em Math. Mech. Solids}, 16(1):122--139, 2011.

\bibitem{Raj_GumMetal14}
K.~R. Rajagopal.
\newblock On the nonlinear elastic response of bodies in the small strain
  range.
\newblock {\em Acta Mech.}, 225(6):1545--1553, 2014.

\bibitem{RAJAGOPAL2018}
K.~R. Rajagopal.
\newblock A note on the linearization of the constitutive relations of
  non-linear elastic bodies.
\newblock {\em Mech. Res. Commun.}, 93:132--137, 2018.
\newblock Mechanics from the 20th to the 21st Century: The Legacy of G{\'e}rard
  A. Maugin.

\bibitem{RAJDENSITY21}
K.~R. Rajagopal.
\newblock An implicit constitutive relation for describing the small strain
  response of porous elastic solids whose material moduli are dependent on the
  density.
\newblock {\em Math. Mech. Solids}, 26(8):1138--1146, 2021.

\bibitem{SAITOETAL03}
T.~Saito, T.~Furuta, J.~H. Hwang, S.~Kuramoto, K.~Nishino, N.~Suzuki, R.~Chen,
  A.~Yamada, K.~Ito, Y.~Seno, T.~Nonaka, H.~Ikehata, N.~Nagasako, C.~Iwamoto,
  Y.~Ikuhara, and T.~Sakuma.
\newblock Multifunctional alloys obtained via a dislocation-free plastic
  deformation mechanism.
\newblock {\em Science}, 300:464--467, 2003.

\bibitem{SAKetal04}
N.~Sakaguch, M.~Niinomi, and T.~Akahori.
\newblock Tensile deformation of ti-nb-ta-zr biomedical alloys.
\newblock {\em Mater. Trans.}, 45:1113--1119, 2004.

\bibitem{SAKAGUCHI2005}
N.~Sakaguchi, M.~Niinomi, T.~Akahori, J.~Takeda, and H.~Toda.
\newblock Effect of {T}a content on mechanical properties of
  {T}i–30{N}b–{XT}a–5{Z}r.
\newblock {\em Mater. Sci. Eng. C}, 25(3):370--376, 2005.
\newblock Selected Papers Presented at the Materials Science and Technology
  2004 Meeting: Titanium for Biomedical, Dental, and Healthcare Applications.

\bibitem{Stoppelli1954}
F.~Stoppelli.
\newblock Un teorema di esistenza ed unicit\`{a} relativo alle equazioni
  dell'elastostatica isoterma per deformazioni finite.
\newblock {\em Ricerche Mat.}, 3:247--267, 1954.

\bibitem{Stoppelli1955}
F.~Stoppelli.
\newblock Sulla sviluppabilit\`{a} in serie di potenze di un parametro delle
  soluzioni delle equazioni dell'{E}lastostatica isoterma.
\newblock {\em Ricerche Mat.}, 4:58--73, 1955.

\bibitem{TALLING2008669}
R.~J. Talling, R.~J. Dashwood, M.~Jackson, S.~Kuramoto, and D.~Dye.
\newblock Determination of (c11-c12) in ti–36nb–2ta–3zr–0.3o (wt.%)
  (gum metal).
\newblock {\em Scr. Mater.}, 59(6):669--672, 2008.

\bibitem{TruesdellNollNLFT}
C.~Truesdell and W.~Noll.
\newblock {\em The Nonlinear Field Theories of Mechanics}.
\newblock Springer-Verlag, Berlin, second edition, 1992.

\bibitem{VAJMURRAJ22}
B.~Vajipeyajula, P.~Murru, and K.~R. Rajagopal.
\newblock Stress concentration due to an elliptic hole in a porous elastic
  plate.
\newblock {\em Math. Mech. Solids}, 28:854--869, 2023.

\bibitem{Wertheim1847}
M.~Wertheim.
\newblock M\'{e}moire sur l'\'{e}lastiqu\'{e} et la coh\'{e}sion des principaux
  tissus du corps humain.
\newblock {\em Ann. Chim. Phys.}, 21:385--414, 1847.

\bibitem{WITHEY200826}
E.~Withey, M.~Jin, A.~Minor, S.~Kuramoto, D.~C. Chrzan, and J.~W. Morris.
\newblock The deformation of “gum metal” in nanoindentation.
\newblock {\em Mat. Sci. Eng. A}, 493(1):26--32, 2008.
\newblock Mechanical Behavior of Nanostructured Materials, a Symposium Held in
  Honor of Carl Koch at the TMS Annual Meeting 2007, Orlando, Florida.

\bibitem{ZHANG2009733}
S.~Q. Zhang, S.~J. Li, M.~T. Jia, Y.~L. Hao, and R.~Yang.
\newblock Fatigue properties of a multifunctional titanium alloy exhibiting
  nonlinear elastic deformation behavior.
\newblock {\em Scr. Mater.}, 60(8):733--736, 2009.

\end{thebibliography}
\bigskip

\centerline{\scshape K. R. Rajagopal}
\smallskip
{\footnotesize
	\centerline{Department of Mechanical Engineering, Texas A\&M University}
	
	\centerline{College Station, TX 77843, USA}
	
	\centerline{\email{krajagopal@tamu.edu}}
}

\bigskip

\centerline{\scshape C. Rodriguez}
\smallskip
{\footnotesize
	\centerline{Department of Mathematics, University of North Carolina}
	
	\centerline{Chapel Hill, NC 27599, USA}
	
	\centerline{\email{crodrig@email.unc.edu}}
}

\end{document}